\documentclass[letterpaper]{article} 
\usepackage{aaai25}  
\usepackage{times}  
\usepackage{helvet}  
\usepackage{courier}  
\usepackage[hyphens]{url}  
\usepackage{graphicx} 
\usepackage{natbib}  
\usepackage{caption} 
\usepackage{algorithm}
\usepackage{algorithmic}
\usepackage{amsmath,amssymb,amsfonts}
\usepackage{multirow}
\usepackage{multicol}
\usepackage{amsthm}
\usepackage{bm}
\usepackage{booktabs}
\usepackage{marvosym}
\usepackage{color}

\usepackage[most]{tcolorbox}
\usepackage{xcolor}
\newtcbox{\myshadowbox}[1][]{colframe=black!75!white, colback=yellow!10, boxrule=1pt, arc=4mm, auto outer arc, drop shadow, #1}

\usepackage{newfloat}
\usepackage{listings}
\DeclareCaptionStyle{ruled}{labelfont=normalfont,labelsep=colon,strut=off} 
\floatstyle{ruled}
\newfloat{listing}{tb}{lst}{}
\floatname{listing}{Listing}
\title{Active Large Language Model-based Knowledge Distillation for Session-based Recommendation}
\author{
    Yingpeng Du\textsuperscript{\rm 1},
    Zhu Sun\textsuperscript{\rm 2}\thanks{Corresponding authors.},
    Ziyan Wang\textsuperscript{\rm 1},
    Haoyan Chua\textsuperscript{\rm 1},
    Jie Zhang\textsuperscript{\rm 1},
    Yew-Soon Ong\textsuperscript{\rm 1,\rm 3}
}
\affiliations{
    \textsuperscript{\rm 1}College of Computing and Data Science, Nanyang Technological University, Singapore\\
    \textsuperscript{\rm 2}Information Systems Technology and Design, Singapore University of Technology and Design, Singapore\\
    \textsuperscript{\rm 3}Department of Computer Science and Engineering, Yonsei University, South Korea\\
    dyp1993@pku.edu.cn,zhu\_sun@sutd.edu.sg,ziyan,wang1753@e.ntu.edu.sg,\\haoyan001@e.ntu.edu.sg,zhangj@ntu.edu.sg,asysong@ntu.edu.sg\\

}

\usepackage{bibentry}

\begin{document}
\nocopyright
\maketitle

\begin{abstract}
Large language models (LLMs) provide a promising way for accurate session-based recommendation (SBR), but they demand substantial computational time and memory. Knowledge distillation (KD)-based methods can alleviate these issues by transferring the knowledge to a small student, which trains a student based on the predictions of a cumbersome teacher. However, these methods encounter difficulties for \textit{LLM-based KD in SBR}. 1) It is expensive to make LLMs predict for all instances in KD. 2) LLMs may make ineffective predictions for some instances in KD, e.g., incorrect predictions for hard instances or similar predictions as existing recommenders for easy instances. In this paper, we propose an active LLM-based KD method in SBR, contributing to sustainable AI. To efficiently distill knowledge from LLMs with limited cost, we propose to extract a small proportion of instances predicted by LLMs. Meanwhile, for a more effective distillation, we propose an active learning strategy to extract instances that are as effective as possible for KD from a theoretical view. Specifically, we first formulate gains based on potential effects (e.g., effective, similar, and incorrect predictions by LLMs) and difficulties (e.g., easy or hard to fit) of instances for KD. Then, we propose to maximize the minimal gains of distillation to find the optimal selection policy for active learning, which can largely avoid extracting ineffective instances in KD. Experiments on real-world datasets show that our method significantly outperforms state-of-the-art methods for SBR.\end{abstract}

%
\section{Introduction}

Recently, large language models (LLMs) have shown the potential to equip recommender systems (RSs) with their extensive knowledge and powerful reasoning capabilities \cite{zhao2023survey}. Most of the existing methods employ LLMs as recommenders based on sophisticated prompt engineering \cite{gao2023retrieval} and fine-tuning with domain-specific knowledge \cite{wu2023survey}.
However, employing LLMs as recommenders usually demands substantial computational time and memory, leading to a high latency and computation requirement during the serving time and limiting real-world applications (e.g., on-device session-based recommendation (SBR) \cite{xia2022device}).  

Intuitively, existing knowledge distillation (KD)-based recommendation methods \cite{tang2018ranking,lee2019collaborative} seem to be potential solutions to these issues, which applied a model-agnostic KD technique to reduce model size while preserving the performance by training a lightweight recommender. Specifically, they train a new compact model (i.e., student recommender) based on the predictions of a well-trained cumbersome model (i.e., teacher recommender).
However, it poses significant challenges for \textit{LLM-based KD} in SBR.  Firstly, \textit{these methods lead to extreme cost, thus restricting the efficiency of KD.}  Specifically, the inference of LLMs is more expensive w.r.t. time and resources than conventional recommenders. Predicting for all instances  (e.g., all sessions in SBR) with LLMs is not a sustainable solution. Secondly, \textit{LLMs may make ineffective (e.g., similar and incorrect) predictions for some instances thus hindering the effectiveness of KD.} Specifically, LLMs may make similar predictions as conventional recommenders in some instances, therefore these instances provide limited information to KD and potentially lead to a redundant distillation. In addition, some instances are hard to correctly predict by LLMs (i.e., failing to model users' real preferences), which may mislead the student recommender.




To address the first challenge, we propose only extracting a small proportion of instances to be predicted by the LLM teacher. Compared to most of the existing KD methods, our strategy does not demand LLM predicting for all instances, which can significantly reduce the cost of KD. To address the second challenge, we propose to conduct active learning to elicit effective LLM-based KD. Specifically, we first define the difficulty of instances based on whether their embedded knowledge can be easily learned by the recommender. Typically, an easy instance may be less informative (e.g., low gain) for LLM-based KD, because its knowledge is likely to be already captured by the student recommender. Although a hard instance may contribute informative knowledge, such an instance may also easily lead to an incorrect prediction (e.g., negative gain) by LLMs when it contains much noise. Meanwhile, we also identify the three potential effects of instances when they are effectively, similarly, and incorrectly predicted by the LLM in ranking tasks, abbreviated as effective/similar/incorrect instances. Then, we formulate their gains with consideration of both their difficulties and effects. Finally, we maximize the minimal gains of distillation to collect effective instances by finding the optimal selection policy for KD. We theoretically and empirically prove that the proposed active learning strategy can select effective instances for KD, that is, selecting as informative as possible instances while avoiding incorrect or similar predictions by LLMs. Our method is a tailored and sustainable solution to LLM-based KD for SBR.

\textbf{Contribution.} In summary, we propose an \underline{A}ctive \underline{L}LM-based \underline{K}nowledge \underline{D}istillation  \underline{Rec}ommendation method, named ALKDRec, as a sustainable yet effective solution to SBR. For more efficient LLM-based KD, we propose to elicit student learning from a small proportion of instances. To ensure effective LLM-based KD theoretically, we propose to maximize the minimal gains of distillation by selecting effective instances for KD. To the best of our knowledge, our method is the first trial to theoretically and practically distill knowledge from LLMs to enhance lightweight recommenders with active learning for efficiency and effectiveness purposes. Experiments on real-world datasets show that our method significantly outperforms state-of-the-art KD methods with representative recommendation backbones. 

\section{Related Work}
\textbf{LLM-based recommendation}.
Leveraging LLMs for recommendation has attained popularity recently due to their advanced capabilities \cite{wu2023survey}. 
Typically, two paradigms have been widely applied, i.e., LLM-as-extractor and LLM-as-recommender. The former paradigm leverages the LLM as a feature extractor for downstream recommendation tasks, employing LLMs to capture and infer user profiles~\cite{zheng2023generative,du2024enhancing}, item descriptions~\cite{ren2024representation,wei2024llmrec}, and other textual data~\cite{liu2024once,geng2022recommendation}. However, LLMs in these methods are not directly designed for recommendation, which may lead to a gap between extracted knowledge and recommendation tasks. 
The latter paradigm adopts the LLM as a predictor to directly provide recommendations for users. They usually employ prompt techniques such as in-context learning~\cite{hou2024large,dai2023uncovering}, prompt tuning~\cite{sun2024large,kang2023llms}, and parameter fine-tuning~\cite{shi2023preliminary,luo2024integrating} to trigger LLMs directly generate recommendations. However, these methods demand substantial computational time and memory, which leads to high latency and increased computation requirements during the serving time. 


\smallskip\noindent\textbf{Knowledge distillation in recommendation}.
Recently, Knowledge distillation \cite{kang2024knowledge} has attracted considerable attention in recommendation, which contributes to the low latency and decreased computation requirement. Most of these methods first train a large teacher recommender and then transfer the knowledge from the teacher to the target compact student recommender, which mainly utilizes soft labels of teacher predictions for KD \cite{lee2019collaborative,pan2019novel,kang2021topology,huang2023aligning}. 
Furthermore, advanced KD technologies are adopted for recommendation such as adversarial distillation \cite{chen2018adversarial,wang2018kdgan}, bidirectional distillation \cite{kweon2021bidirectional,lee2021dual}, privileged feature distillation \cite{wang2021privileged,xu2020privileged,zhang2020distilling,liu2023future,deng2023bkd}, debias distillation \cite{chen2023unbiased}, multi-modal/model distillation \cite{liu2023semantic,wei2024multi,kang2023distillation,sun2024distillation}, and sparsity distillation \cite{wang2024dynamic}, and LLM-based distillation \cite{wang2024can,liu2024large,cui2024distillation}. However, these methods are not tailored for efficient LLM-based KD due to (1) extreme costs with LLM inference (e.g., fine-tuning or distilling knowledge based on all instances); and (2) ineffective predictions by LLMs for KD in SBR.


\section{The Proposed Method}

\begin{figure*} \centering
 \includegraphics[width=0.8\textwidth]{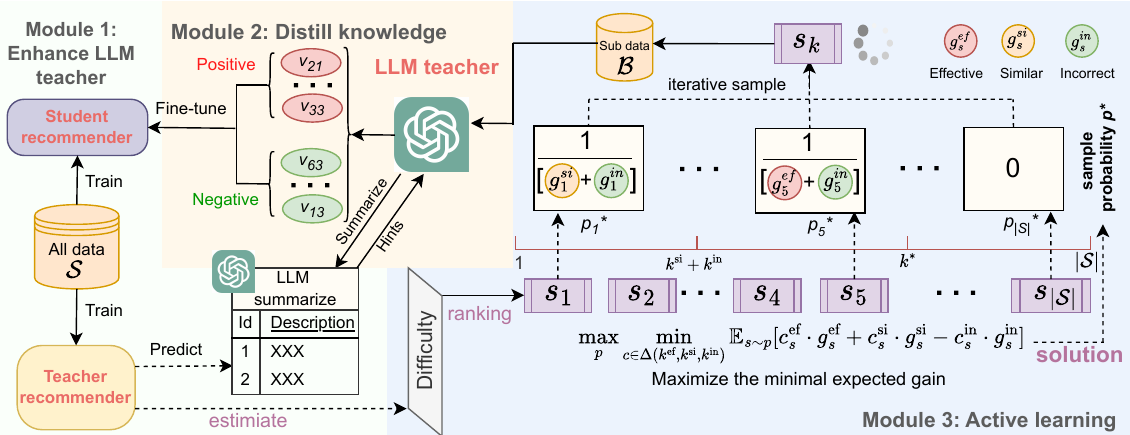}
\caption{The architecture of the ALKDRec method.} \label{Method}
\end{figure*}

\textbf{Problem Definition}.
Let $\mathcal{V} = \{v_1, v_2, ..., v_{N}\}$ represent the item set with $N$ items. Each anonymous session $s\in\mathcal{S}$ records the user's interaction with $l$ items, which is denoted as $s = \{v_1^s, v_2^s, ..., v_l^s|v_i^s \in \mathcal{V}\}$. We assume to know the titles of items. Given $s$, the task of SBR is to rank all items that have not interacted yet and provide a top-$K$ item list for the next item $v_{l+1}^s$ prediction. Our goal is to design a KD method in SBR, enabling the student recommender to effectively and efficiently learn from the LLM teacher with small-scale parameters (e.g., low dimension in latent space) and achieve accurate recommendations.

\smallskip\noindent\textbf{Model Overview}.
 The overall architecture of the proposed method is shown in Figure \ref{Method}. Firstly, to make the \textit{LLM teacher} align with domain-specific knowledge, we elicit it to summarize from predictions of the \textit{conventional teacher recommender} that is trained by domain-specific user behaviors (Module 1). Secondly, to distill knowledge from the LLM teacher, we trigger the \textit{student recommender} to learn from its prediction (Module 2). Finally, to efficiently distill knowledge from the LLM teacher, we propose an active learning strategy to select effective instances for KD theoretically and practically (Module 3).



\subsection{Enhance LLM teacher with conventional teacher recommender}
LLMs usually lack domain-specific knowledge (e.g., user behavior patterns) related to the SBR tasks, which may hinder the effective KD of LLMs. To make LLMs align with domain-specific knowledge, we elicit the LLM to make summarization with the help of the conventional teacher recommender {\small$T_\text{rec}(\cdot)$} which is trained on all instances $\mathcal{S}$, and can capture the collaborative signals behind users' behaviors. Then, we elicit the LLM teacher {\small$T_\text{llm}(\cdot)$} to summarize users' behavior patterns using prompt {\small$prompt_\text{smr}$} as hints by analyzing predictions of the teacher recommender, 
\begin{equation*}
      {\small  Hints = T_\text{llm}(prompt_\text{smr},\{rank^{rec}_{s_1},\cdots,rank^{rec}_{s_L}|s_i\in\mathcal{S}\}),}
\end{equation*}

\noindent where {\small $rank^{rec}_{s} = T_\text{rec}(s)$} denotes the predictions of teacher recommender on session $s$. The detailed prompt {\small$prompt_\text{smr}$} for \underline{s}um\underline{m}a\underline{r}ization is shown in Appendix A. With summarized hints encapsulated with domain-specific knowledge, we can trigger the LLM teacher to make recommendations as a ranking list  {$\small rank_{s}^{llm} = T_\text{llm}(Prompt_\text{dst},s, Hints,\mathcal{C}_s),$} where $prompt_\text{dst}$ is the prompt that triggers {\small$T_\text{llm}(\cdot)$} for SBR shown in Appendix A. {\small$\mathcal{C}_s = rank_{s}^{rec}[1:\kappa]$} is the candidate set of $\kappa$ (e.g., $\kappa=50$) top items suggested by {\small$T_\text{rec}(\cdot)$}.

\subsection{Distill knowledge from the LLM teacher}
To distill knowledge from the LLM teacher into the student, we fine-tune the student recommender $S_\text{rec}(\cdot)$ that is well-trained based on all instances $\mathcal{S}$, thus making the student recommender learn to imitate the prediction of the LLM teacher, i.e., $rank_s^{llm}$, with pair-wise loss, 
\begin{equation}\label{distill}
{\small max\sum_{s\in\mathcal{S}} \sum_{v\in rank_s^{llm}}\alpha_{v}\cdot\log\sigma(S_\text{rec}(s,v)-S_\text{rec}(s,v'))},
\end{equation}

\noindent where $S_\text{rec}(s,v)$ denotes the predicted score for item $v$ in session $s$ by the student recommender, and $v'$ denotes the negative item. $\alpha_v$ denotes the weight of pair-wise loss w.r.t. positive sample $v$, where we assign the item $v$ to the top position with a higher weight.  

Existing KD strategies which learn from teachers regarding all instances $\mathcal{S}$ are extremely costly for LLM-based KD, as the inference of LLM teachers is much more expensive than conventional teachers. To this end, we propose to distill knowledge from a subset of all instances, i.e., $\mathcal{B} \subset \mathcal{S}$ where $\vert\mathcal{B}\vert<<\vert\mathcal{S}\vert$. However, how to select an effective instance subset is crucial for LLM-based KD in SBR.

\subsection{Active learning for LLM-based distillation}

\subsubsection{Difficulty of instances in distillation.} For efficient LLM-based KD, it is intuitive to query the LLM teacher with informative instances, the patterns of which are hardly captured by conventional recommenders. To be specific, if the knowledge of an instance can be easily captured by conventional recommenders, we regard it as an \textbf{easy instance} for LLM-based KD. Hence, employing the LLM teacher distill on these easy instances is likely to lead to redundant distillation. Contrarily, the \textbf{hard instances} may provide informative knowledge for KD, but it is hard to predict correctly by the LLM teacher. To measure the difficulty of instances, we propose to assess the negative consistency between the session representation and each item within the session in the conventional teacher recommender, i.e., ${\small df_s = - \frac{1}{|s|}\sum\nolimits_{v\in s}\sigma(\text{Encode}(s)^T\cdot \boldsymbol{e}_{v}),}$
where ${\small\text{Encoder}(s) \in \mathbb{R}^d}$ and ${\small \boldsymbol{e}_{v}\in\mathbb{R}^d}$ are the embeddings of session $s$ and item $v$, respectively. Thus, a large negative consistency (i.e., greater difficulty) indicates that the teacher recommender can hardly capture the knowledge. 

\subsubsection{Effects of instances in distillation.} 

Instances may show different effects in distillation when distilling knowledge from the LLM teacher, which can be categorized as follows.
 \textbf{\textit{Effective instances}} indicate that the LLM teacher makes accurate predictions but their patterns are hardly captured by conventional recommenders, thus usually providing informative knowledge to KD.
\textbf{\textit{Similar instances}} indicate that the LLM teacher and the student recommender share similar predictions, which usually provide limited informative knowledge to KD. 
 \textbf{\textit{Incorrect instances}} indicate that LLM makes incorrect predictions on users' behaviors, which may mislead the student model training.

\subsubsection{Expected gain of distillation.}
To comprehensively model the gain of instances on KD, we propose considering the difficulty and effect of the instance. Effective instances usually provide informative knowledge to KD. Therefore, we assume the gain of the effective instance $g_s^\text{ef}$ is proportional to the difficulty of the instance $df_s$. Similar instances usually provide limited informative knowledge to KD, which can improve the confidence of extracted knowledge thus showing marginal gains to KD, i.e., $g_s^\text{si}<g_s^\text{ef}$. Incorrect instances may mislead student model training and show the negative gain to KD, i.e., ${\small -g_s^\text{in}<0}$ and ${\small|g_s^\text{in}|<g_s^\text{ef}}$. The detailed implementation of these gains can be found in the \textit{Implementation Details} of Section \textit{Experiments}.

 To elicit effective LLM-based KD, we propose an active learning strategy to select instances with the maximal expected gain, given by,
\begin{equation}\label{eq1}
    {\small\max\nolimits_{p}\mathbb{E}_{s\sim p}[c_{s}^\text{ef}\cdot g_s^\text{ef} + c_{s}^\text{si}\cdot g_s^\text{si} - c_{s}^\text{in}\cdot g_s^\text{in}]},
\end{equation}
where $p\in\mathbb{R}^{\vert\mathcal{S}\vert}$ denotes the probability of selecting one specific instance for active KD. Without loss of generality, we rank instances based on their effective gains, i.e., $g_1^\text{ef}>g_2^\text{ef}>\cdots>g_{\vert\mathcal{S}\vert}^\text{ef}$. In addition, $c_{s}^\text{ef}=1$, $c_{s}^\text{si}=1$, or $c_{s}^\text{in}=1$ represent the indicator, indicating an effective, similar, or incorrect instance for KD, satisfying $c_{s}^\text{ef} + c_{s}^\text{si} + c_{s}^\text{in}=1$.

 In real-world scenarios, it is hard to identify the indicator value of an instance $s$ before eliciting the LLM teacher to predict for it, making it intractable to directly optimize the expected gain in Equation (\ref{eq1}) for LLM-based KD. To this end, we propose maximizing the expected gain of distillation regardless of the specific cases of all instances, which can be achieved by maximizing the minimal expected gain for the instance selection inspired by \cite{baykalrobust}.

\subsubsection{Maximize the minimal gains of distillation.}
 As it is intractable to obtain specific values of indicator of all instances $\{c_{1},\cdots,c_{\vert\mathcal{S}\vert}\}$, we only assume to know the overall number for each type of instances, i.e., {\small$\sum_s c_{s}^\text{ef} =k^\text{ef}$} for effective instances, {\small$\sum_s c_{s}^\text{si} = k^\text{si}$} for similar instances, and {\small$\sum_s c_{s}^\text{in} = k^\text{in}$} for incorrect instances, where {\small$k^\text{ef} + k^\text{si} + k^\text{in} = \vert\mathcal{S}\vert$}. We formulate to maximize the minimal expected gain, 
\begin{equation}\label{eq2}
{\small
    \max_{p}\min_{c\in\Delta(k^\text{ef},k^\text{si},k^\text{in})}\underbrace{\mathbb{E}_{s\sim p}[c_{s}^\text{ef}\cdot g_s^\text{ef} + c_{s}^\text{si}\cdot g_s^\text{si} - c_{s}^\text{in}\cdot g_s^\text{in}]}_{=\sum_s p_s\cdot[c_{s}^\text{ef}\cdot g_s^\text{ef} + c_{s}^\text{si}\cdot g_s^\text{si} - c_{s}^\text{in}\cdot g_s^\text{in}]:=z(p,c)}},
\end{equation}
where {\small$\Delta(k^\text{ef},k^\text{si},k^\text{in})$} denotes a set assembling all combinations of instances which satisfies {\small$\sum_s c_s^t = k^t$ for $t\in \{\text{ef},\text{si},\text{in}\}$}, and probability $p$ is the active strategy for the instance selection.

To solve the max-min problem in Equation (\ref{eq2}), we provide two theorems to prove the lower bound and upper bound of the expected gain. Then, we illustrate that the equilibrium point is the exact solution for this problem\footnote{Due to page limitation, we only show the main sketch of proofs, where the detailed proofs can be found in the Appendix B.}.

\theoremstyle{definition}
\newtheorem{definition}{Definition}

\begin{definition}[Probability $p^*$] Suppose we have a vector $p^* = (p_1^*,\cdots,p_{\vert\mathcal{S}\vert}^*)$ satisfying 

\begin{equation*}
{\small
p_s^*=\left\{
\begin{aligned}
&1/({H_{k^*}(g_s^\text{si}+g_s^\text{in})})  , &\text{if } 1 \le s\le k^\text{si} + k^\text{in}, \\
&1/({H_{k^*}(g_s^\text{ef}+g_s^\text{in})}) , &\text{if } k^\text{si} + k^\text{in}<s\le k^*,\\
&0  , &\text{if } s> k^*,\\
\end{aligned}
\right.}
\end{equation*}
where 
\begin{equation*}
{\small
H_k=\left\{
\begin{aligned}
&\sum\nolimits_{s=1}^{k}\frac{1}{g_s^\text{si}+g_s^\text{in}}  , &\text{if } 1 \le k \le k^\text{si} + k^\text{in}, \\
&H_{k^\text{si} + k^\text{in}}+ \sum\nolimits_{s=1}^{k}\frac{1}{g_s^\text{ef}+g_s^\text{in}} , &\text{if } k^\text{si} + k^\text{in}< k \le k^*.\\
\end{aligned}
\right.}
\end{equation*}
\end{definition}
\noindent We have $p^*$ is a probability distribution as {\small$p_s^*\ge 0$} and
{\small
$\sum\nolimits_{s=1}^{|\mathcal{S}|} p_s^* = \frac{1}{H_{k^*}}[\sum\nolimits_{s=1}^{k^\text{si} + k^\text{in}} \frac{1}{g_s^\text{si}+g_s^\text{in}} +\sum\nolimits_{  s = k^\text{si} + k^\text{in}+1 }^{k^*}\frac{1}{g_s^\text{ef}+g_s^\text{in}}] = 1.$}

To illustrate the property of the probability $p^*$, we prove its relation to a lower bound of the expected gain $z(p^*,c)$ in Equation (\ref{eq2}), which is detailed as follows.

\newtheorem{theorem}{Theorem}

\begin{theorem}\label{lm1}
The expected gain $z(p,c)$ has a constant lower bound with $p = p^*$, i.e.,
\begin{equation*}
    {\small\min_{c\in\Delta(k^\text{ef},k^\text{si},k^\text{in})}\mathbb{E}_{s\sim p^*}[c_{s}^\text{ef}\cdot g_s^\text{ef} + c_{s}^\text{si}\cdot g_s^\text{si} - c_{s}^\text{in}\cdot g_s^\text{in}] = \Gamma(k^*)}
\end{equation*}
where {\small$\Gamma(k^*) = [k^\text{ef}+k^\text{si}-N+G_{k^*}]/H_{k^*}$} and
{\small\begin{equation*}
G_k\!=\!\left\{
\begin{aligned}
&\sum\nolimits_{s=1}^{k} g_s^\text{si} /(g_s^\text{si}+g_s^\text{in})  , &\text{if } 1 \le k \le k^\text{si} + k^\text{in}, \\
&G_{k^\text{si} + k^\text{in}} \!+\! \sum\nolimits_{s=1}^{k}g_s^\text{ef}/(g_s^\text{ef}+g_s^\text{in})  , &\text{if } k^\text{si} + k^\text{in}< k \le k^*.\\
\end{aligned}
\right.
\end{equation*}}
\end{theorem}

\renewcommand{\proofname}{Proof sketch}

\begin{proof}

Assigning {\small$p=p^*$} in expected gain {\small$z(p,c)$}, we can reformulate the expected gain {\small$z(p^*,c)$} as
{\small\begin{align*}
   & z(p^*,c)   = \frac{1}{H_{k^*}}\Bigl[ \sum\nolimits_{s=1}^{ k^\text{si}+k^\text{in}} \left(c_{s}^\text{ef}\frac{g_s^\text{ef}+g_s^\text{in}}{g_s^\text{si}+g_s^\text{in}} + c_{s}^\text{si}\cdot 1 - \frac{g_s^\text{in}}{g_s^\text{si} + g_s^\text{in}}\right) \\
    & \quad\quad   + \sum\nolimits_{s= k^\text{si}+k^\text{in}+1}^{ k^* } \left(c_{s}^\text{ef}\cdot 1 + c_{s}^\text{si}\frac{g_s^\text{si}+g_s^\text{in}}{g_s^\text{ef}+g_s^\text{in}} - \frac{g_s^\text{in}}{g_s^\text{ef} + g_s^\text{in}}\right) \Bigl]. 
\end{align*}}

To minimize the gain {\small$z(p^*,c)$}, we should assign effective instances ({\small$c_{s}^\text{ef}=1$}) into {\small$[k^\text{si} + k^\text{in},|\mathcal{S}|]$}, similar instances ({\small$c_{s}^\text{si}=1$}) into {\small$[1, k^\text{si}+k^\text{in}]$}, and incorrect instances into rest place of {\small$[1,|\mathcal{S}|]$}. Therefore, we have the minimal gain {\small$z(p^*,c)$} as 
{\small\begin{align*}
    &\min_{c\in\Delta(k^\text{ef},k^\text{si},k^\text{in})} z(p^*,c)  = \frac{1}{H_{k^*}}\Bigl[0 + k^\text{si} + (k^* - k^\text{si} - k^\text{in}) + 0 \\ 
       & - \sum\nolimits_{s=1}^{k^\text{si}+k^\text{in}}\frac{g_s^\text{in}}{g_s^\text{si} + g_s^\text{in}}  - \sum\nolimits_{s= k^\text{si}+k^\text{in}+1}^{k^*}\frac{g_s^\text{in}}{g_s^\text{ef} + g_s^\text{in}} \Bigl] \\
       & = \frac{1}{H_{k^*}}[k^\text{ef}+k^\text{si}-N+G_{k^*}],
\end{align*}}
where  {\small$\Gamma(k^*) = [k^\text{ef}+k^\text{si}-N+G_{k^*}]/H_{k^*}$}.
\end{proof}

To explore the gain $z(p,c)$ w.r.t. the combination $c$, we make a specific definition of $\hat{c}$ as follows.
\begin{definition}[Combination $\hat{c}^\text{ef}$, $\hat{c}^\text{si}$, and $\hat{c}^\text{in}$] Suppose we have three vectors {\small$ \hat{c}^\text{t} = (\hat{c}_1^\text{t},\cdots,\hat{c}_{|\mathcal{S}|}^\text{t})$} for ${\small t\in\{\text{ef},\text{si},\text{in}\}}$ as 
\begin{equation*}{\small
\hat{c}^\text{ef}_s\!=\!\left\{
\begin{aligned}
&0  , &\text{if } 1 \le s\le k^\text{si} + k^\text{in}, \\
&\frac{\Gamma(k^*)+g_s^\text{in}}{g_s^\text{si}+g_s^\text{in}}  , &\text{if } k^\text{si} + k^\text{in}<s\le k^*,\\
&\Bigl[k^\text{ef}\!-\!\!\!\!\!\sum_{s = k^\text{si} \!+\! k^\text{in}}^{k^*}\!\!\!\!\frac{\Gamma(k^*)\!+\!g_s^\text{in}}{g_s^\text{si}\!+\!g_s^\text{in}}\Bigl]/(|\mathcal{S}|\!-\!k^*) , &\text{if } s > k^*,\\
\end{aligned}
\right.}
\end{equation*}
\begin{equation*}
{\small
\hat{c}^\text{si}_s\!=\!\left\{
\begin{aligned}
& \frac{\Gamma(k^*)+g_s^\text{in}}{g_s^\text{ef}+g_s^\text{in}}, &\text{if } 1 \le s\le k^\text{si} + k^\text{in}, \\
& 0 , &\text{if } k^\text{si} + k^\text{in}<s\le k^*,\\
&\Bigl[k^\text{si}\!-\!\!\!\sum_{s = 1}^{k^\text{si} \!+\! k^\text{in}}\!\!\frac{\Gamma(k^*)\!+\!g_s^\text{in}}{g_s^\text{ef}+g_s^\text{in}}\Bigl]/(|\mathcal{S}|\!-\!k^*) , &\text{if } s > k^*,\\
\end{aligned}
\right.}
\end{equation*}
and ${\small \hat{c}^\text{in}_s = 1 - \hat{c}^\text{ef}_s - \hat{c}^\text{si}_s}$. 

We have ${\small \hat{c}:=(\hat{c}^\text{ef}, \hat{c}^\text{si}, \hat{c}^\text{in})\in \Delta(k^\text{ef},k^\text{si},k^\text{in})}$ that is a feasible combination as ${\small\sum_{s\le |\mathcal{S}|}\hat{c}^\text{ef}_s = k^\text{ef}}$, ${\small\sum_{s\le |\mathcal{S}|}\hat{c}^\text{si}_s = k^\text{si}}$, and {\small$\sum_{s\le |\mathcal{S}|}\hat{c}^\text{in}_s= k^\text{in}$} are established. 
\end{definition}

To illustrate the property of the combination $\hat{c}$, we prove that its relation to an upper bound of the expected gain $z(p,\hat{c})$ in Equation (\ref{eq2}), which is detailed as follows.
\begin{theorem}\label{lm2}
The expected gain $z(p,c)$ has a constant upper bound with $c = \hat{c}$, i.e.,
\begin{equation*}
  {\small
  \max_{p}\mathbb{E}_{s\sim p}[\hat{c}_{s}^\text{ef}\cdot g_s^\text{ef} + \hat{c}_{s}^\text{si}\cdot g_s^\text{si} - \hat{c}_{s}^\text{in}\cdot g_s^\text{in}] = \Gamma(k^*)}
\end{equation*} 
\begin{equation*}
{\small
k^* = \text{argmax}_{s=1,\cdots,|\mathcal{S}|}\frac{g_s^\text{ef}+k^\text{ef}+k^\text{si}-|\mathcal{S}|}{H_s} }
\end{equation*}
 where {\small$\Gamma(k^*) = [k^\text{ef}+k^\text{si}-N+G_{k^*}]/H_{k^*}$ and $g_1>\cdots>g_{\vert\mathcal{S}\vert}$}.
\end{theorem}
\begin{proof}
Assign $c= \hat{c}$ in gain $z(p,c)$, we have 
{\small\begin{equation*}
     z(p, \hat{c}) \! =\! \sum_{s=1}^{k^*} p_s  \Gamma(k^*)  + \sum_{s=k^*}^{|\mathcal{S}|} p_s  \underbrace{[\hat{c}_{s}^\text{ef} (g_s^\text{ef}+g_s^\text{in}) + \hat{c}_{s}^\text{si} (g_s^\text{si}+g_s^\text{in}) - g_s^\text{in}] }_{:=t(c_{s}^\text{ef},c_{s}^\text{si},s)}
\end{equation*}}
Furthermore, we reformulate $t(c_{s}^\text{ef},c_{s}^\text{si},s)$ when $s\ge k^*$, and prove that $t(c_{s}^\text{ef},c_{s}^\text{si},s) \le g_s^\text{ef}$ by inequality scaling.

According to the definition of $k^*$,  we have 
{\small\begin{align*}
 &\Gamma(k^*) \!\ge\! \Gamma(k^*\!+\!1) \!=\! \frac{\!k^\text{ef}\!+\!k^\text{si}\!-\!N\!+\!G_{k^*}+g_{k^*+1}/(g_{k^*+1}+l_{k^*+1})}{H_{k^*}+1/(g_{k^*+1}+l_{k^*+1})}\\
 &{\ge}\frac{g_{k^*+1}/(g_{k^*+1}+l_{k^*+1})}{1/(g_{k^*+1}+l_{k^*+1})} = g_{k^*+1} \ge g_s^\text{ef}  \text{ for } (s \ge k^* +1)
\end{align*}}
Therefore, we have ${\small\Gamma(k^*)\ge g_s^\text{ef} \ge t(c_{s}^\text{ef},c_{s}^\text{si},s) \text{ when } s\ge k^*.}$

To maximize $z(p, \hat{c})$, we apparently have $\sum_{s= 1}^{k^*} p_s= 1$ and $p_s=0$ when $s>k^*$. Therefore, we have the conclusion 
        $\max_{p}\mathbb{E}_{s\sim p}[\hat{c}_{s}^\text{ef}\cdot g_s^\text{ef} + \hat{c}_{s}^\text{si}\cdot g_s^\text{si} - \hat{c}_{s}^\text{in}\cdot g_s^\text{in}] = \Gamma(k^*)$
\end{proof}
According to the above theorems, we can achieve the maximum of minimal expected gain according to the following theorem.
\begin{theorem}\label{lm3}
The maximum of the minimal expected gain $z(p,c)$ in Equation (\ref{eq2}) can be achieved as $\Gamma(k^*)$ in this max-min game when ${\small p^* = (p_1^*,\cdots,p_{\vert\mathcal{S}\vert}^*)}$  with {\small$ k^* = \text{argmax}_{s=1,\cdots,|\mathcal{S}|}{[g_s^\text{ef}+k^\text{ef}+k^\text{si}-|\mathcal{S}|]}/{H_s}$}  as 
{\small
\begin{align*}
    &\Gamma(k^*) = \min_{c\in(k^\text{ef},k^\text{si},k^\text{in})}\mathbb{E}_{s\sim p^*}[c_{s}^\text{ef}\cdot g_s^\text{ef} + c_{s}^\text{si}\cdot g_s^\text{si} - c_{s}^\text{in}\cdot g_s^\text{in}] \\
    &\le \max_{p}\min_{c\in(k^\text{ef},k^\text{si},k^\text{in})}\mathbb{E}_{s\sim p}[c_{s}^\text{ef}\cdot g_s^\text{ef} + c_{s}^\text{si}\cdot g_s^\text{si} - c_{s}^\text{in}\cdot g_s^\text{in}]\\
    & \le \max_{p} \mathbb{E}_{s\sim p}[\hat{c}_{s}^\text{ef}\cdot g_s^\text{ef} + \hat{c}_{s}^\text{si}\cdot g_s^\text{si} - \hat{c}_{s}^\text{in}\cdot g_s^\text{in}] = \Gamma(k^*).
\end{align*}}
\end{theorem}
To gain the subset $\mathcal{B} \subset \mathcal{S}$ for KD, we select instances one by one where each instance is sampled based on the active learning policy $p^*$ in Definition 1 with $k^*$ shown in Theorem 2, the detailed . In summary, we provide an active learning strategy (i.e., active instance selection policy $p^*$) that can maximize the minimal expected gain for the instance selection, contributing to the effective LLM-based KD with limited cost. Compared to \cite{baykalrobust} which focused on binary classification tasks, the proposed active learning strategy in ALKDRec can model the complicated cases for ranking tasks of SBR, i.e., the LLM teacher may make effective, similar, and incorrect predictions for KD in SBR tasks.


\subsection{Experiments and Analysis}
We conduct extensive experiments to evaluate the performance of ALKDRec and answer four research questions. \textbf{RQ1}: Whether ALKDRec outperforms state-of-the-art methods for SBR w.r.t. effectiveness and efficiency?
 \textbf{RQ2}: Whether ALKDRec benefits from the knowledge distilled from the LLM teacher for SBR?
 \textbf{RQ3}: Whether ALKDRec benefits from the active learning strategy for SBR?
 \textbf{RQ4}: How do hyper-parameters in ALKDRec affect its performance for SBR?

\subsubsection{Datasets.} We evaluate ALKDRec and baselines on two real-world datasets, namely Hetrec2011-ML and Amazon-Games. The former dataset contains user ratings of the online movie service MovieLens. The latter dataset is the ‘Video Game’ category of the Amazon dataset, which contains users' reviews on games. All items in these two datasets are attached with title information. For each user, we extract her interactions within one day as a session sorted by timestamps, and we also filter out sessions with less than 5 records. After data preprocessing, we have $12,323$ sessions, $8,475$ items, and $112,034$ interactions for Hetrec2011-ML, and $20,091$ sessions, $26,138$ items, and $132,677$ interactions for Amazon-Games.

\renewcommand{\arraystretch}{1.25}
\begin{table*}[tp]
  \centering
  \fontsize{6.5}{6.8}\selectfont
  \caption{Performance of different methods.  * indicates statistically significant improvement of the proposed method to baseline models on t-test ($p < 0.05$). `Improve' indicates the relative improvements of ALKDRec over the strongest baseline. }
  \label{tab:RSComparison}
\begin{tabular}{c|c|c|cc|cccccc|c|c}\midrule\midrule
Datasets                        & \multicolumn{1}{l|}{Backbone} & Metric    & Student Rec & Teacher Rec & DE            & FTD              & HTD           & unKD    & DSL           & DLLM2Rec      & ALKDRec      & Improve.        \\\midrule\midrule
\multirow{12}{*}{Hetrec2011-ML} & \multirow{4}{*}{FPMC}          & recall@5  & 0.00730     & 0.01339     & 0.00690       & \underline{ 0.00933}    & 0.00893       & 0.00527 & 0.00527       & 0.00933       & \textbf{0.01258*} & 34.78\%         \\
                                &                                & ndcg@5    & 0.00379     & 0.00781     & 0.00358       & 0.00457          & \underline{ 0.00570} & 0.00361 & 0.00356       & 0.00467       & \textbf{0.00747*} & 31.10\%         \\
                                &                                & recall@10 & 0.01055     & 0.02191     & 0.01258       & \underline{ 0.01866}    & 0.01460       & 0.00893 & 0.00893       & 0.01623       & \textbf{0.02150*} & 15.22\%         \\
                                &                                & ndcg@10   & 0.00415     & 0.01002     & 0.00658       & \underline{ 0.00901}    & 0.00723       & 0.00426 & 0.00436       & 0.00731       & \textbf{0.01019*} & 13.06\%         \\\cline{2-13}
                                & \multirow{4}{*}{STAMP}         & recall@5  & 0.00609     & 0.00852     & 0.00852       & 0.00811          & 0.00609       & 0.00487 & \underline{ 0.00852} & 0.00609       & \textbf{0.01014*} & 19.05\%         \\
                                &                                & ndcg@5    & 0.00364     & 0.00446     & \underline{ 0.00529} & 0.00450          & 0.00337       & 0.00309 & 0.00452       & 0.00364       & \textbf{0.00670*} & 26.76\%         \\
                                &                                & recall@10 & 0.01177     & 0.01907     & \underline{ 0.01501} & 0.01298          & 0.01217       & 0.01014 & 0.01298       & 0.01217       & \textbf{0.01623*} & 8.11\%          \\
                                &                                & ndcg@10   & 0.00563     & 0.00937     & \underline{ 0.00729} & 0.00616          & 0.00517       & 0.00509 & 0.00542       & 0.00575       & \textbf{0.00799*} & 9.58\%          \\\cline{2-13}
                                & \multirow{4}{*}{AttMix}        & recall@5  & 0.00122     & 0.00527     & 0.00527       & \textbf{0.00649} & 0.00325       & 0.00365 & 0.00527       & 0.00527       & \underline{0.00609 }         & -6.25\%         \\
                                &                                & ndcg@5    & 0.00067     & 0.00361     & 0.00306       & \underline{ 0.00332}    & 0.00222       & 0.00235 & 0.00306       & 0.00284       & \textbf{0.00370*} & 11.40\%         \\
                                &                                & recall@10 & 0.00406     & 0.00933     & 0.00933       & 0.01014          & 0.00730       & 0.00568 & 0.00933       & \underline{ 0.01055} & \textbf{0.01095*} & 3.84\%          \\
                                &                                & ndcg@10   & 0.00180     & 0.00419     & 0.00409       & 0.00366          & 0.00351       & 0.00262 & 0.00406       & \underline{ 0.00448} & \textbf{0.00547*} & 22.14\%         \\\midrule\midrule
\multirow{12}{*}{Amazon-Games}  & \multirow{4}{*}{FPMC}          & recall@5  & 0.03309     & 0.05773     & 0.02886       & 0.03260          & \underline{ 0.03459} & 0.00921 & 0.03211       & 0.03409       & \textbf{0.03583*} & 3.60\%          \\
                                &                                & ndcg@5    & 0.01857     & 0.03322     & 0.01471       & 0.01869          & 0.01917       & 0.00504 & 0.01802       & \underline{ 0.01954} & \textbf{0.02188*} & 11.98\%         \\
                                &                                & recall@10 & 0.05101     & 0.08311     & 0.04703       & 0.05126          & \underline{ 0.05200} & 0.01643 & 0.04331       & 0.05126       & \textbf{0.05748*} & 10.53\%         \\
                                &                                & ndcg@10   & 0.02199     & 0.03688     & 0.01938       & \underline{ 0.02251}    & 0.02199       & 0.00696 & 0.01891       & 0.02225       & \textbf{0.02610*} & 15.97\%         \\\cline{2-13}
                                & \multirow{4}{*}{STAMP}         & recall@5  & 0.03583     & 0.05375     & 0.02937       & 0.03359          & 0.03085       & 0.01643 & 0.03086       & \underline{ 0.03608} & \textbf{0.03707} & {2.76\%} \\
                                &                                & ndcg@5    & 0.02025     & 0.03091     & 0.01624       & 0.01938          & 0.01689       & 0.00850 & 0.01777       & \underline{ 0.02051} & \textbf{0.02103} & {2.55\%} \\
                                &                                & recall@10 & 0.04902     & 0.07066     & 0.04231       & 0.05051          & 0.04877       & 0.02290 & 0.04505       & \underline{ 0.05051} & \textbf{0.05325*} & 5.42\%          \\
                                &                                & ndcg@10   & 0.02119     & 0.03104     & 0.01724       & \underline{ 0.02246}    & 0.02215       & 0.00934 & 0.01905       & 0.02188       & \textbf{0.02342*} & 4.26\%          \\\cline{2-13}
                                & \multirow{4}{*}{AttMix}        & recall@5  & 0.01916     & 0.04379     & 0.01095       & 0.00325          & 0.01568       & 0.00406 & 0.00487       & \underline{ 0.01866} & \textbf{0.01916} & {2.67\%} \\
                                &                                & ndcg@5    & 0.01144     & 0.02546     & 0.00600       & 0.00206          & 0.00910       & 0.00234 & 0.00280       & \underline{ 0.01103} & \textbf{0.01118} & {1.33\%} \\
                                &                                & recall@10 & 0.02264     & 0.05200     & 0.01717       & 0.00690          & 0.02140       & 0.00649 & 0.00933       & \underline{ 0.02215} & \textbf{0.02936*} & 32.59\%         \\
                                &                                & ndcg@10   & 0.00983     & 0.02315     & 0.00687       & 0.00323          & 0.00977       & 0.00310 & 0.00355       & \underline{ 0.00987} & \textbf{0.01343*} & 36.07\%        \\\midrule\midrule
\end{tabular}
\end{table*}


\subsubsection{Evaluation Protocol and Metrics.} We randomly split sessions into training, validation, and test sets by 6:2:2. For evaluation phase, we adopt the widely used leave-one-out evaluation protocol. Specifically for each session, we hold out the last interacted item for testing. We adopt two widely used evaluation metrics for top-$K$ recommendation:  recall and normalized discounted cumulative gain (ndcg) \cite{sun2022daisyrec}, where $K$ is set as $5/10$ empirically. Experimental results are recorded as the average of five runs.

\subsubsection{Backbone models and baselines}
 Inspired by prior setting \cite{kang2020rrd}, we test KD methods across 3 representative backbone recommenders for SBR. \textbf{FPMC} \cite{rendle2010factorizing} adopts the Markov Chain to capture sequential patterns of users' behaviors.  \textbf{STAMP}  \cite{liu2018stamp} proposes to model users' current interests from the short-term memory of the last clicks. \textbf{AttMix} \cite{zhang2023efficiently} models multi-level reasoning over item transitions from both concept-view and instance-view. For baseline methods, we take the following state-of-the-art \textit{agnostic} KD methods as the baselines for comparison, including \textbf{DE} \cite{kang2020rrd}, \textbf{FTD}\cite{kang2021topology}, \textbf{HTD} \cite{kang2021topology},  \textbf{unKD} \cite{chen2023unbiased}, \textbf{DSL} \cite{wang2024dynamic}, and \textbf{DLLM2Rec} \cite{cui2024distillation}.


\subsubsection{Implementation Details.}  For an effective instance, the low difficulty implies a less informative contribution to KD in SBR. Therefore, we assign the gain value of $g_s^\text{ef}$ based on its difficulty with an exponential distribution {\small$g_s^\text{ef} = 1/[rank(df_s)]^{\mu}$, where $rank(\cdot)$} denotes the ascending ranking of instance among all instances w.r.t. difficulty, and we set $\mu=10$ empirically. For similar and incorrect instances, we assign them with lower gain values compared to effective instances, i.e., $g_s^\text{si}=g_s^\text{in}= g_s^\text{ef}/2$. 

 We adopt the GPT-4-turbo as the LLM teacher to distill knowledge from $500$ instances for all datasets and backbones except $750$ instances for FPMC in Amazon-Games. We set the number of effective/similar/incorrect instances as 1:5:4 for all datasets by the grid search. For $\alpha_{v}$ in Equation (\ref{distill}), we assign 3/2/1 when the item $v$ in ranking position $[1,5]/(5,15]/(15,25]$. Following \citet{lee2019collaborative}, we set the latent dimensions for the teacher and student recommenders at 100 and 10 across all KD methods for a fair comparison. We set the learning rate as $1\times 10^{-3}$ with Adam optimizer and batch size as $1024$ for all methods.  For baseline methods, we set their parameters as authors’ implementation if they exist, otherwise we tune them to their best. Our code can be found in the Supplementary Material.

\begin{figure} \centering
 \includegraphics[width=0.5\textwidth]{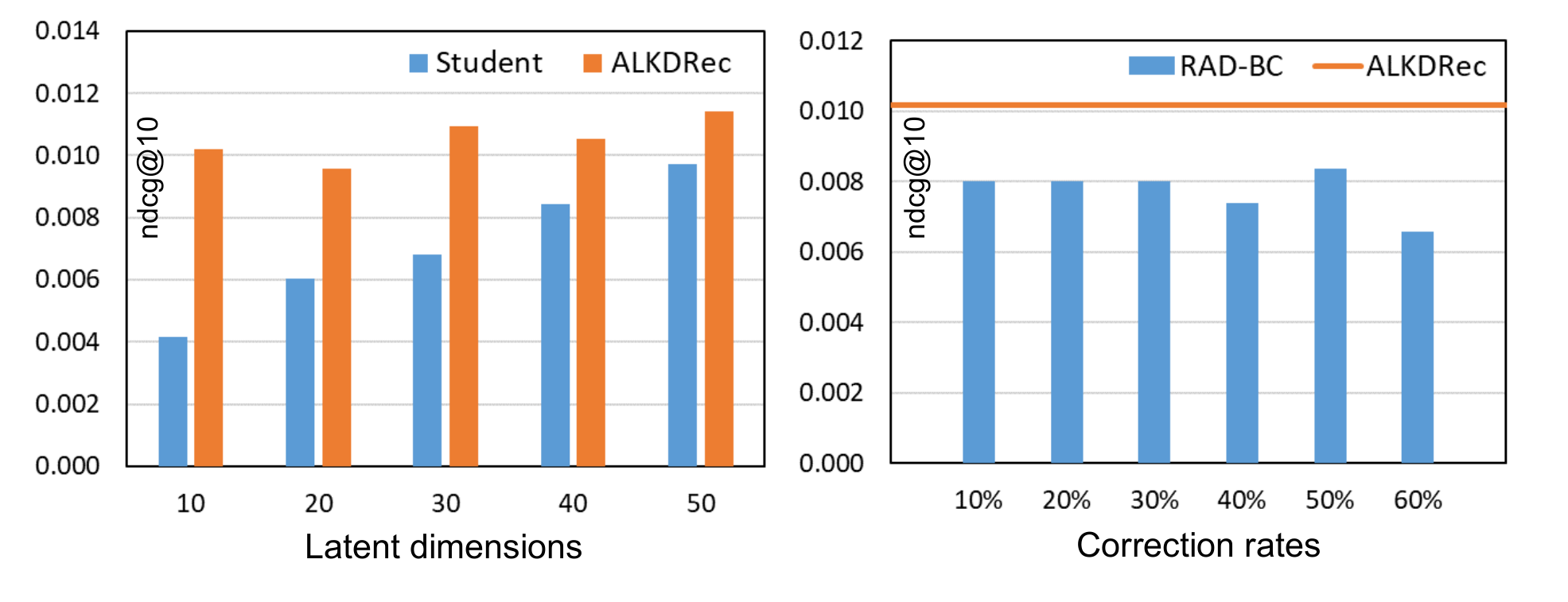}
\caption{ Performance (ndcg@10) with different (a) latent dimensions, and (b) correction rates.} \label{case}
\end{figure}

\subsection{Performance Comparison (RQ1)}
\textbf{Effectiveness.} Table \ref{tab:RSComparison} presents the performance of various methods, including the student recommenders, teacher recommenders, and KD recommendation methods. To highlight the most notable results, we have bolded the best results and underlined the runner up results for KD-based methods on each dataset.  From the experimental results, we can get the following conclusions.
      Firstly, ALKDRec significantly outperforms all KD-based baselines in most cases, which shows the effectiveness of ALKDRec in KD for SBR. Specifically, ALKDRec achieves average improvements of 17.03\%, 9.81\%, and 12.97\% when compared to the best KD-based baselines on the three backbones FPMC, STAMP, and AttMix, respectively.
      Second, KD-based baselines outperform student recommender in most cases, which shows the necessity of KD for SBR. However, they may degrade students' performance sometimes, e.g., STAMP and AttMix perform worse than vanilla students on Amazon-Games. This may be attributed to their limited model size making it hard to capture the complicated patterns, which demand strong generalization capability with large model size, e.g., multi-level reasoning in AttMix. Therefore, distilling this complicated knowledge is daunting for the student and leads to degradation. This also explains why the complicated AttMix backbone performs worse than other backbones with a limited model size.  Third, DLLM2Rec performs well among KD-based baselines but underperforms our ALKDRec, indicating the necessity of our active learning strategy for LLM-based KD. 
      
\smallskip\noindent\textbf{Efficiency}. In the training phase, our ALKDRec only extracts from a subset of instances (e.g., around 44 minutes and 8.6 USD for ChatGPT API in Amazon-Games), which is much faster and computation-saving than training on all samples (e.g., around 1782 minutes and 347.0 USD for ChatGPT API in Amazon-Games). In addition, ALKDRec with a limited model size even outperforms the teacher recommender with a $10X$ model size, indicating that distilling knowledge from LLMs is a promising way with efficient solutions for SBR.

\subsection{Ablation Study (RQ2\&3)}
To assess the effectiveness of ALKDRec's module design, we compare with several variants.
\textbf{(1) TR} distills knowledge from \textbf{T}eacher \textbf{R}ecommender by mimicking its predictions, which fine-tunes the well-trained student recommender for KD. 
\textbf{(2) Random} distills knowledge from LLM by randomly sampling a subset of instances.
\textbf{(3) Hardest/Easiest} distill knowledge from LLM for the hardest/easiest instances for their difficulties.
\textbf{(4) RAD-BC} replaces our active learning strategy in ALKDRec with the \textbf{r}obust \textbf{a}ctive \textbf{d}istillation on \textbf{b}inary \textbf{c}lassification tasks \cite{baykalrobust} that only models the correct and incorrect instances. We set the number of correct instances in RAD-BC as the total number of effective/similar instances in ALKDRec.  Table \ref{table_ablation} shows the performance of ALKDRec and its variant, indicating the following conclusions.

\begin{table}[tp]
  \centering
  \fontsize{7}{7}\selectfont
  \caption{Performance of the variants for ablation studies measured by recall (R@10) and ndcg (N@10). }
  \label{table_ablation}
\begin{tabular}{|c|cc|cc|cc|}
        \midrule \multicolumn{7}{|c|}{Hetrec2011-ML}    \\\midrule
Backbone & \multicolumn{2}{c|}{FPMC}            & \multicolumn{2}{c|}{STAMP}           & \multicolumn{2}{c|}{AttMix}          \\\midrule
Metric     & R@10        & N@10          & R@10        & N@10          & R@10        & N@10          \\\midrule
TR         & 0.01907          & 0.00897          & 0.01055          & 0.00476          & 0.00852          & 0.00375          \\
Random     & 0.01907          & 0.00856          & 0.01339          & 0.00755          & 0.00974          & 0.00469          \\
Easiest    & 0.01623          & 0.00693          & 0.01542          & 0.00754          & 0.00933          & 0.00426          \\
Hardest    & 0.01298          & 0.00639          & 0.01582          & 0.00756          & 0.00933          & 0.00439          \\
RAD-BC   & 0.01420          & 0.00658          & 0.01177          & 0.00601          & \textbf{0.01136} & 0.00519          \\
ALKDRec  & \textbf{0.02150} & \textbf{0.01019} & \textbf{0.01623} & \textbf{0.00799} & 0.01095          & \textbf{0.00547} \\\midrule\midrule
\multicolumn{7}{|c|}{Amazon-Games}         \\\midrule
Backbone& \multicolumn{2}{c|}{FPMC}            & \multicolumn{2}{c|}{STAMP}           & \multicolumn{2}{c|}{AttMix}          \\\midrule
Metric     & R@10        & N@10          & R@10        & N@10          & R@10        & N@10          \\\midrule
TR         & 0.05175          & 0.02238          & 0.05200          & 0.02223          & 0.02264          & 0.01002          \\
Random     & 0.05101	      & 0.02239          & 0.05225          & 0.02203          & 0.02538          & 0.01132          \\
Easiest    & 0.05076	      & 0.02238         & 0.05026          & 0.02161          & 0.01966          & 0.00905          \\
Hardest    & 0.04877	      & 0.02133          & 0.05001          & 0.02173          & 0.02239          & 0.00997          \\
RAD-BC   & 0.05200	      & 0.02321          & 0.05225          & 0.02228          & 0.02712          & 0.01210          \\
ALKDRec  & \textbf{0.05748} & \textbf{0.02610} & \textbf{0.05325} & \textbf{0.02342} & \textbf{0.02936} & \textbf{0.01343} \\\midrule         
\end{tabular}
\end{table}

 \textbf{RQ2}: Firstly, ALKDRec consistently outperforms TR, which shows the necessity of distilling knowledge from the LLM teacher. Secondly, Random shows improvements over TR in most cases, confirming the potential of LLM-based KD. In addition, Figure \ref{case} (a) shows the performance of ALKDRec with FPMC varying with different latent dimensions. It indicates that ALKDRec can consistently improve the student recommenders regardless of their sizes, illustrating its generalizability and robustness for LLM-based KD.

\textbf{RQ3}: ALKDRec outperforms the other instance selection strategies (i.e., Random, Easiest, Hardest, and RAD-BC) in most cases, indicating the effectiveness of LLM-based KD by maximizing the minimal gains of distillation for SBR. Firstly, simply selecting the easiest and hardest samples performs badly among these methods. To be specific, the easiest instance leads to redundant (less informative) distillation as its knowledge is likely to be already captured by the student recommender, contributing less to KD for SBR. The hardest instance makes it hard to be correctly predicted by the LLMs, which may mislead the training of student recommenders. Secondly, ALKDRec outperforms RAD-BC in most cases, which is attributed to ALKDRec can avoid both similar and incorrect predictions of LLMs to KD for SBR. Specifically, Figure \ref{case} (b) compares the performances between ALKDRec and RAD-BC with varying correction rates \cite{baykalrobust} when adopting backbone FPMC on Hetrec2011-ML. It shows that ALKDRec consistently outperforms RAD-BC, indicating the superiority of the proposed active learning strategy for LLM-based KD.

\subsection{Hyper-Parameter Study (RQ4)}\label{sec_hyper}

We now investigate the impacts of key parameters for ALKDRec, including the ratios of effective/similar/incorrect instances and the number of sampled instances for KD $\vert\mathcal{B}\vert$. Figure \ref{hyper1} shows the performance of ALKDRec with FPMC across various combinations of effective and similar ratios. It demonstrates that the optimal performance is achieved with effective, similar, and incorrect instances at a ratio of $1:5:4$. Figure \ref{hyper2} shows the performance of ALKDRec with FPMC across different numbers of sampled instances. Specifically, we observe that the performance first improves with more instances, but begins to degrade when exceeding a threshold. We suggest selecting $500$ instances for active KD learning in real-world applications for effective and efficient consideration.

\begin{figure} \centering
 \includegraphics[width=0.5\textwidth]{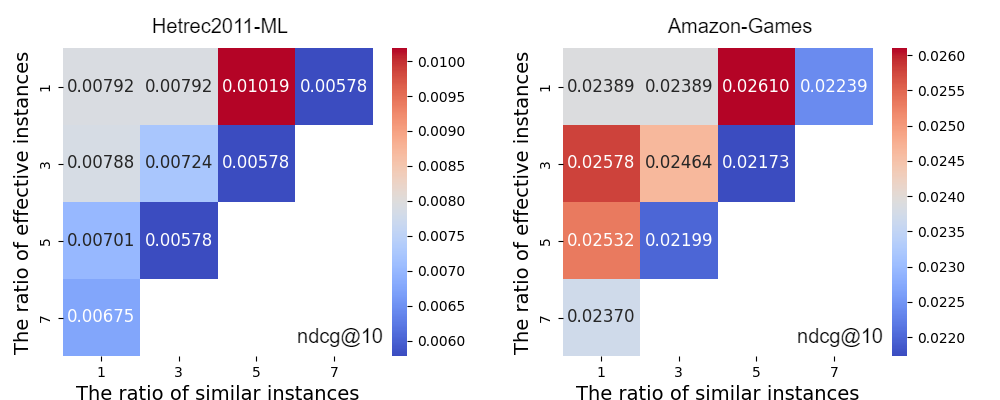}
\caption{ Performance (ndcg@10) of ALKDRec across various effective instance ratios and similar case ratios.} \label{hyper1}
\end{figure}
\begin{figure} \centering
 \includegraphics[width=0.5\textwidth]{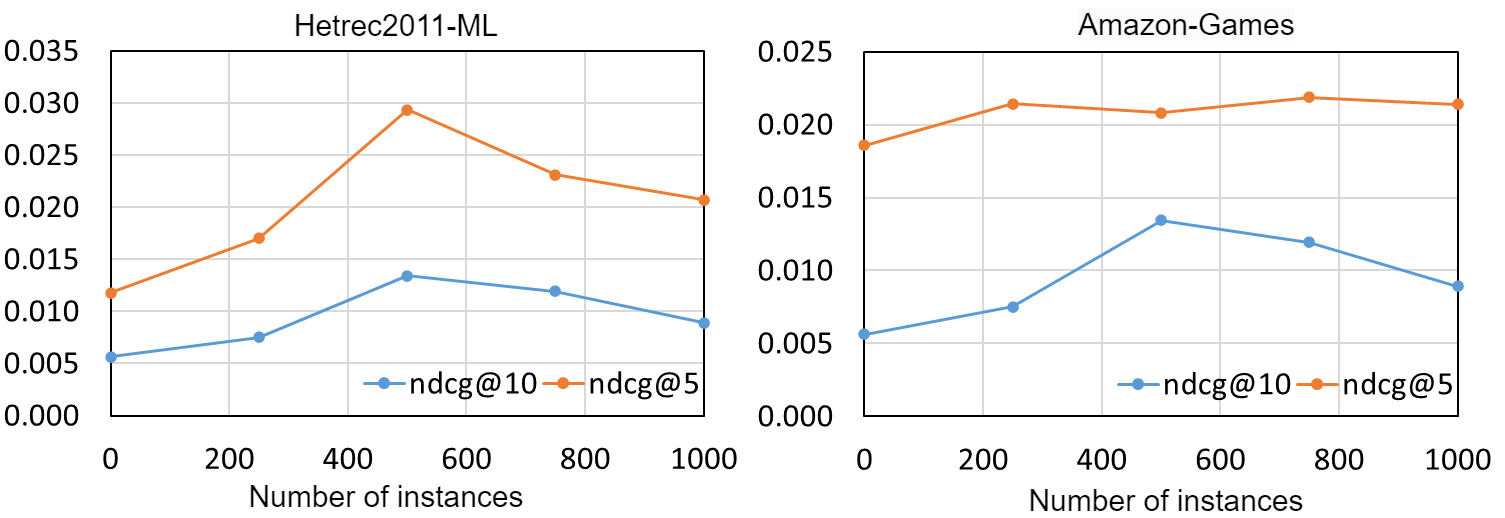}
\caption{ Performance (ndcg@5/10) of ALKDRec vary with different instance numbers for KD.} \label{hyper2}
\end{figure}

\section{Conclusion}
In this paper, we propose an {A}ctive {L}LM-based KD {Rec}ommendation method for a sustainable and effective solution to SBR. For more efficient LLM-based KD, we propose to elicit student learning from a small proportion of instances. To theoretically ensure effective LLM-based KD, we propose maximizing the minimal gains of distillation by selecting effective instances for KD. Experiments on real-world datasets show that our method significantly outperforms state-of-the-art KD methods with representative recommendation backbones. In future work, we will explore how to efficiently distill embedding-level knowledge from LLMs to the student recommenders.

\bibliography{aaai25}

\onecolumn
\appendix
\setcounter{theorem}{0} 
\setcounter{definition}{0} 

\section{A. Prompts for the LLMs teacher}

\begin{tcolorbox}[breakable, colback=yellow!10, colframe=black!75!white, title= $Prompt_\text{dst}$: Triggers the LLM
teacher for recommendation Recommendation Prompt]
\textit{{\color{red} Question}: You are an AI recommender system, please make accurate recommendations to the user according to his/her behaviors. 
There are some hints for recommendation: {\color{blue}  \{ Hint 1, Hint 2, $\cdots$ \} } 
\newline
The user has bought/watched {\color{blue}\{\%d\}} movies, whose titles are namely {\color{blue}\{\%s\}} (e.g., $<$ID1:Casino Royale$>$  $<$ID2: Batman$>$ ).
Based on these interactions, recommend a ranking list of items for the user from a candidate set: {\color{blue}\{\%s\}}.
\newline
Specifically, rank ALL these candidate items and provide EXACTLY 25 items for recommendations, where the item in the top position indicates the higher recommendation intent.
        Just output WITHOUT ANY OTHER MESSAGE: [$<$ID1$>$,...,$<$ID25$>$]  with $<$ID$?>$ surrounding JUST the ID WITHOUT title and split all IDs with a comma.
\newline {\color{red}ChatGPT}: {\color{blue} Top-$25$ recommendation results.} \newline  \newline 
\textbf{Example Output}:\newline
{\small[$<$ID5130$>$,$<$ID6022$>$,$<$ID7304$>$,$\cdots$,$<$ID1698$>$]}
}
\end{tcolorbox}

\vspace{20pt}

\begin{tcolorbox}[breakable, colback=yellow!10, colframe=black!75!white, title=$Prompt_\text{smr}$: Summarize users’ behavior patterns.]
\textit{{\color{red} Question}: You are an AI assistant, please STRICTLY summarize the main logic of the recommendation results provided by MY recommender system based on users' behaviors. Specifically, the users have interacted with several movies/games (i.e., users' behaviors). The recommender system provides Top-20 items (i.e., recommendation results) based on users' behaviors, where the item in the top position indicates the higher recommendation intent.  We provide {\color{blue}\{\%d\}} cases for you to summarize the relationship between users' behaviors and recommendation results:
\newline
For case {\color{blue}\{\%d\}} (e.g., $1,2,3\cdots$), the user's interactions are {\color{blue}\{\%s\}}. The recommendation results are {\color{blue}\{\%s\}}.
\newline
Do not output the detailed analysis of specific items, and make high-level summarization JUST WITH PROVIDED CASES, e.g., 
\newline 1. Ignore diversity in genres or series. For example, Although the user already owns this one, there would still be interest in the same one, e.g., its the special version.\newline 
 \newline {\color{red}ChatGPT}: {\color{blue}  \{ Hint 1, Hint 2, $\cdots$ \} } \newline  \newline 
\textbf{Example Output}:\newline
\underline{Sequel and Franchise Continuity:} The recommender system often suggests sequels or other movies within the same franchise when a user has interacted with a previous installment. \newline  
\underline{Genre Affinity:} The system tends to recommend movies within the same or similar genres to those previously interacted with. \newline  
\underline{Casting and Crew Commonality:} Recommendations frequently feature movies that have common directors, actors, era or release year, or production companies with films the user has previously enjoyed. \newline  
\underline{In disregard of diversity in genres or series}. For example, Although the user already owns the movie, there would still be interest in the same movie, e.g., its the special version. \newline
}
\end{tcolorbox}

\vspace{5pt}

\section{B: Detailed Proof of theorems}
\subsection{Proof of Theorem 1.}
\theoremstyle{definition}
\begin{definition}[Probability $p^*$] Suppose we have a vector $p^* = (p_1^*,\cdots,p_{|\mathcal{S}|}^*)$ satisfying 

\begin{equation*}
p_s^*=\left\{
\begin{aligned}
&\frac{1}{H_{k^*}(g_s^\text{si}+g_s^\text{in})} , &\text{if } 1 \le s\le k^\text{si} + k^\text{in}, \\
&\frac{1}{H_{k^*}(g_s^\text{ef}+g_s^\text{in})} , &\text{if } k^\text{si} + k^\text{in}< s\le k^*,\\
&0  , &\text{if } s > k^*,\\
\end{aligned}
\right.
\end{equation*}
where 
\begin{equation*}
H_k=\left\{
\begin{aligned}
&\sum_{s=1}^{k}\frac{1}{g_s^\text{si}+g_s^\text{in}} , &\text{if } 1 \le k \le k^\text{si} + k^\text{in}, \\
&H_{k^\text{si} + k^\text{in}}+ \sum_{s=1}^{k}\frac{1}{g_s^\text{ef}+g_s^\text{in}} , &\text{if } k^\text{si} + k^\text{in}< k \le k^*.\\
\end{aligned}
\right.
\end{equation*}

\noindent We have $p^*$ is a probability distribution because $p_s^*\ge 0$ and
\begin{align*}
&\sum_{s=1}^{|\mathcal{S}|} p_s^* = \sum_{s=1}^{k^\text{si} + k^\text{in}} p_s^* +\sum_{  s = k^\text{si} + k^\text{in}+1 }^{k^*} p_s^*\frac{1}{H_{k^*}}\Bigl[\sum_{s=1}^{k^\text{si} + k^\text{in}} \frac{1}{g_s^\text{si}+g_s^\text{in}} +\sum_{  s = k^\text{si} + k^\text{in}+1 }^{k^*}\frac{1}{g_s^\text{ef}+g_s^\text{in}}\Bigl]\frac{H_{k^*}}{H_{k^*}} = 1.
\end{align*}
\end{definition}

To illustrate the property of the probability $p^*$, we prove its relation to a lower bound of the expected gain $z(p^*,c)$ in Equation (2), which is detailed as follows:

\begin{theorem}\label{lm1}
The expected gain $z(p,c)$ has a constant lower bound with $p = p^*$, i.e.,

\begin{equation*}
   \min_{c\in\Delta(k^\text{ef},k^\text{si},k^\text{in})}\mathbb{E}_{s\sim p^*}[c_{s}^\text{ef}\cdot g_s^\text{ef} + c_{s}^\text{si}\cdot g_s^\text{si} - c_{s}^\text{in}\cdot g_s^\text{in}] = \Gamma(k^*)
\end{equation*}
where $\Gamma(k^*) = \frac{k^\text{ef}+k^\text{si}-|\mathcal{S}|+G_{k^*}}{H_{k^*}}$ and
\begin{equation*}
G_k=\left\{
\begin{aligned}
&\sum_{s=1}^{k} \frac{g_s^\text{si}}{g_s^\text{si}+g_s^\text{in}}  , &\text{if } 1 \le k \le k^\text{si} + k^\text{in}, \\
&G_{k^\text{si} + k^\text{in}}+ \sum_{s=1}^{k} \frac{g_s^\text{ef}}{g_s^\text{ef}+g_s^\text{in}}  , &\text{if } k^\text{si} + k^\text{in}< k \le k^*.\\
\end{aligned}
\right.
\end{equation*}
\end{theorem}

\begin{proof}

By assigning $p=p^*$ in the expected gain $z(p,c)$, we have 
\begin{align*}
   & z(p^*,c) = \sum_{s=1}^{|\mathcal{S}|} p_s^* \cdot [c_{s}^\text{ef}\cdot g_s^\text{ef} + c_{s}^\text{si}\cdot g_s^\text{si} - c_{s}^\text{in}\cdot g_s^\text{in}] \\
      & = \sum_{s=1}^{|\mathcal{S}|} \frac{1}{H_{k^*}(g_s^\text{si}+g_s^\text{in})}\cdot [c_{s}^\text{ef}\cdot g_s^\text{ef} + c_{s}^\text{si}\cdot g_s^\text{si} - c_{s}^\text{in}\cdot g_s^\text{in}] \\
   & = \sum_{ s=1}^{k^\text{si}+k^\text{in}} \frac{1}{H_{k^*}(g_s^\text{si}+g_s^\text{in})}\cdot [c_{s}^\text{ef}\cdot g_s^\text{ef} + c_{s}^\text{si}\cdot g_s^\text{si} - c_{s}^\text{in}\cdot g_s^\text{in}]  + \sum_{ s = k^\text{si}+k^\text{in}+1}^{ k^* } \frac{1}{H_{k^*}(g_s^\text{ef}+g_s^\text{in})}\cdot [c_{s}^\text{ef}\cdot g_s^\text{ef} + c_{s}^\text{si}\cdot g_s^\text{si} - c_{s}^\text{in}\cdot g_s^\text{in}] \\
   & = \sum_{s=1}^{k^\text{si}+k^\text{in}} \frac{1}{H_{k^*}(g_s^\text{si}+g_s^\text{in})}\cdot [c_{s}^\text{ef} g_s^\text{ef} + c_{s}^\text{si} g_s^\text{si} - g_s^\text{in} + c_{s}^\text{ef}g_s^\text{in} + c_{s}^\text{si}g_s^\text{in}]  + \sum_{ s = k^\text{si}+k^\text{in}+1}^{k^*} \frac{1}{H_{k^*}(g_s^\text{ef}+g_s^\text{in})}\cdot [c_{s}^\text{ef} g_s^\text{ef} + c_{s}^\text{si} g_s^\text{si} - g_s^\text{in} + c_{s}^\text{ef}g_s^\text{in} + c_{s}^\text{si}g_s^\text{in}] \\
    & = \frac{1}{H_{k^*}}\Bigl[\sum_{ s=1}^{k^\text{si}+k^\text{in}} (c_{s}^\text{ef}\frac{g_s^\text{ef}+g_s^\text{in}}{g_s^\text{si}+g_s^\text{in}} + c_{s}^\text{si}\cdot 1 - \frac{g_s^\text{in}}{g_s^\text{si} + g_s^\text{in}})   + \sum_{ s=k^\text{si}+k^\text{in}+1}^{ k^* } (c_{s}^\text{ef}\cdot 1 + c_{s}^\text{si}\frac{g_s^\text{si}+g_s^\text{in}}{g_s^\text{ef}+g_s^\text{in}} - \frac{g_s^\text{in}}{g_s^\text{ef} + g_s^\text{in}})\Bigl] 
\end{align*}
To minimize the gain $z(p^*,c)$, we should first assign $c_{s}^\text{ef} = 1$ when $s\ge k^*$ because $p_s^*=0$ and the coefficients of $c_{s}^\text{ef}$ are always greater than $c_{s}^\text{si}$'s. Then, we should assign the rest $(k^\text{ef} - (N-k^*)) = (k^* - k^\text{si} - k^\text{in})$ effective instances into the interval $[k^\text{si}+k^\text{in},k^*]$, i.e., $c_{s}^\text{ef} = 1$ when $k^\text{si}+k^\text{in}\le s \le k^*$, because the assign effective instances into the interval $[k^\text{si}+k^\text{in},k^*]$ lead to less gain than the interval $[1,k^\text{si}+k^\text{in}]$ due to
\begin{equation*}
    \frac{g_s^\text{ef}+g_s^\text{in}}{g_s^\text{si}+g_s^\text{in}} - 1 =  \frac{g_s^\text{ef} - g_s^\text{si}}{g_s^\text{si}+g_s^\text{in}} > \frac{g_s^\text{ef} - g_s^\text{si}}{g_s^\text{ef}+g_s^\text{in}} = 1 - \frac{g_s^\text{si} +g_s^\text{in}}{g_s^\text{ef}+g_s^\text{in}}.
\end{equation*}
Finally, we assign the redundant instances into the interval $[1, k^\text{si}+k^\text{in}]$ because the interval $[k^\text{si}+k^\text{in},\mathcal{S}]$ is filled with the effective instances. Therefore, the minimal gain $z(p^*,c)$ can be formulated as 
\begin{align*}
    &\min_{c} z(p^*,c)  = \frac{1}{H_{k^*}}\Bigl[0 + k^\text{si} + (k^* - k^\text{si} - k^\text{in}) + 0  - \sum_{ s=1}^{k^\text{si}+k^\text{in}}\frac{g_s^\text{in}}{g_s^\text{si} + g_s^\text{in}}  - \sum_{s= k^\text{si}+k^\text{in}+1}^{ k^* }\frac{g_s^\text{in}}{g_s^\text{ef} + g_s^\text{in}} \Bigl] \\
       & = \frac{1}{H_{k^*}}\Bigl[ k^\text{ef} +k^\text{si}-N + \sum_{ s=1}^{ k^\text{si}+k^\text{in}}\frac{g_s^\text{si}}{g_s^\text{si} + g_s^\text{in}} -\frac{g_s^\text{si} + g_s^\text{in}}{g_s^\text{si} + g_s^\text{in}}  + \sum_{ s= k^\text{si}+k^\text{in}+1}^{ k^* }\frac{g_s^\text{ef}}{g_s^\text{ef} + g_s^\text{in}} - \frac{g_s^\text{ef}+g_s^\text{in}}{g_s^\text{ef} + g_s^\text{in}}  \Bigl] \\
       & = \frac{1}{H_{k^*}}[k^\text{ef}+k^\text{si}-N+G_{k^*}] = \Gamma(k^*)
\end{align*}
\end{proof}

\subsection{Proof of Theorem 2.}
To explore the gain $z(p,c)$ w.r.t. the combination $c$, we make a specific definition of $\hat{c}$ as follows:
\begin{definition}[Combination $\hat{c}^\text{ef}$, $\hat{c}^\text{si}$, and $\hat{c}^\text{in}$] Suppose we have three vectors {\small$ \hat{c}^\text{t} = (\hat{c}_1^\text{t},\cdots,\hat{c}_{|\mathcal{S}|}^\text{t})$} for ${\small t\in\{\text{ef},\text{si},\text{in}\}}$ as 
\begin{equation*}
\hat{c}^\text{ef}_s=\left\{
\begin{aligned}
&0  , &\text{if } 1 \le s\le k^\text{si} + k^\text{in}, \\
&\frac{\Gamma(k^*)+g_s^\text{in}}{g_s^\text{si}+g_s^\text{in}}  , &\text{if } k^\text{si} + k^\text{in}< s\le k^*,\\
&\frac{k^\text{ef}\!-\!\!\sum_{s = k^\text{si} + k^\text{in}}^{k^*}\!\!\frac{\Gamma(k^*)+g_s^\text{in}}{g_s^\text{si}+g_s^\text{in}}}{|\mathcal{S}|-k^*} , &\text{if } s > k^*,\\
\end{aligned}
\right.
\end{equation*}
\begin{equation*}
\hat{c}^\text{si}_s=\left\{
\begin{aligned}
& \frac{\Gamma(k^*)+g_s^\text{in}}{g_s^\text{ef}+g_s^\text{in}}, &\text{if } 1 \le s\le k^\text{si} + k^\text{in}, \\
& 0 , &\text{if } k^\text{si} + k^\text{in}< s\le k^*,\\
&\frac{k^\text{si}\!-\!\sum_{s = 1}^{k^\text{si} + k^\text{in}}\frac{\Gamma(k^*)+g_s^\text{in}}{g_s^\text{ef}+g_s^\text{in}}}{|\mathcal{S}|-k^*} , &\text{if } s > k^*,\\
\end{aligned}
\right.
\end{equation*}
and ${\small \hat{c}^\text{in}_s = 1 - \hat{c}^\text{ef}_s - \hat{c}^\text{si}_s}$. 

We have ${\small \hat{c}:=(\hat{c}^\text{ef}, \hat{c}^\text{si}, \hat{c}^\text{in})\in \Delta(k^\text{ef},k^\text{si},k^\text{in})}$ is a feasible combination as ${\small\sum_{s\le |\mathcal{S}|}\hat{c}^\text{ef}_s = k^\text{ef}}$, ${\small\sum_{s\le |\mathcal{S}|}\hat{c}^\text{si}_s = k^\text{si}}$, and {\small$\sum_{s\le |\mathcal{S}|}\hat{c}^\text{in}_s= k^\text{in}$} are established. 
\end{definition}

To illustrate the property of the combination $\hat{c}$, we prove its relation to an upper bound of the expected gain $z(p,\hat{c})$ in Equation (2), which is detailed as follows.
\begin{theorem}\label{lm2}
The expected gain $z(p,c)$ has a constant upper bound with $c = \hat{c}$, i.e.,
\begin{equation*}
  \max_{p}\mathbb{E}_{s\sim p}[\hat{c}_{s}^\text{ef}\cdot g_s^\text{ef} + \hat{c}_{s}^\text{si}\cdot g_s^\text{si} - \hat{c}_{s}^\text{in}\cdot g_s^\text{in}] = \Gamma(k^*)
\end{equation*}
and 
\begin{equation*}
k^* = \text{argmax}_{s=1,\cdots,|\mathcal{S}|}\frac{g_s^\text{ef}+k^\text{ef}+k^\text{si}-|\mathcal{S}|}{H_s} 
\end{equation*}
 where {\small$\Gamma(k^*) = [k^\text{ef}+k^\text{si}-N+G_{k^*}]/H_{k^*}$ and $g_1>\cdots>g_{|\mathcal{S}|}$}.
\end{theorem}
\begin{proof}
Assign $c= \hat{c}$ in the expected gain $z(p,c)$, we have 
\begin{align*}
    & z(p, \hat{c}) = \sum_{s=1}^{|\mathcal{S}|} p_s \cdot [\hat{c}_{s}^\text{ef}\cdot g_s^\text{ef} + \hat{c}_{s}^\text{si}\cdot g_s^\text{si} - \hat{c}_{s}^\text{in}\cdot g_s^\text{in}] \\
    & = \sum_{s=1}^{|\mathcal{S}|} p_s \cdot [\hat{c}_{s}^\text{ef}\cdot g_s^\text{ef} + \hat{c}_{s}^\text{si}\cdot g_s^\text{si} - (1 - \hat{c}_{s}^\text{ef} - \hat{c}_{s}^\text{si})\cdot g_s^\text{in}] = \sum_{ s=1}^{ |\mathcal{S}|} p_s \cdot [\hat{c}_{s}^\text{ef}\cdot (g_s^\text{ef}+g_s^\text{in}) + \hat{c}_{s}^\text{si}\cdot (g_s^\text{si}+g_s^\text{in}) - g_s^\text{in}] \\
    & = \sum_{s = 1}^{k^\text{ef}+k^\text{si}} p_s \cdot \Bigl[0\cdot (g_s^\text{ef}+g_s^\text{in}) + \frac{\Gamma(k^*)+g_s^\text{in}}{g_s^\text{si}+g_s^\text{in}}\cdot (g_s^\text{si}+g_s^\text{in}) - g_s^\text{in}\Bigl] 
    + \sum_{s = k^\text{ef}+k^\text{si}+1}^{k^*} p_s \cdot \Bigl[\frac{\Gamma(k^*)+g_s^\text{in}}{g_s^\text{si}+g_s^\text{in}}\cdot (g_s^\text{ef}+g_s^\text{in}) + 0 \cdot (g_s^\text{si}+g_s^\text{in}) - g_s^\text{in}\Bigl]\\
    &\quad\quad\quad\quad\quad+ \sum_{s=k^*+1}^{|\mathcal{S}|} p_s \cdot [\hat{c}_{s}^\text{ef}\cdot (g_s^\text{ef}+g_s^\text{in}) + \hat{c}_{s}^\text{si}\cdot (g_s^\text{si}+g_s^\text{in}) - g_s^\text{in}] \\
    & = \sum_{s=1}^{k^*} p_s \Gamma(k^*) + \sum_{s=k^*+1}^{|\mathcal{S}|} p_s \underbrace{[\hat{c}_{s}^\text{ef}\cdot (g_s^\text{ef}+g_s^\text{in}) + \hat{c}_{s}^\text{si}\cdot (g_s^\text{si}+g_s^\text{in}) - g_s^\text{in}] }_{:=t(c_{s}^\text{ef},c_{s}^\text{si},u)}
\end{align*}

Furthermore, we reformulate $t(c_{s}^\text{ef},c_{s}^\text{si},u)$ when $s\ge k^*$, 
\begin{align*}
    & t(c_{s}^\text{ef},c_{s}^\text{si},s) \\
    &= (g_s^\text{ef} + g_s^\text{in}) \cdot \Bigl[k^\text{ef}-\sum_{s = k^\text{si} + k^\text{in}}^{k^*}\frac{\Gamma(k^*)+g_s^\text{in}}{g_s^\text{si}+g_s^\text{in}}\Bigl]/(|\mathcal{S}|-k^*)  + (g_s^\text{si} + g_s^\text{in}) \cdot \Bigl[k^\text{si}-\sum_{s = k^\text{si} + k^\text{in}}^{k^*}\frac{\Gamma(k^*)+g_s^\text{in}}{g_s^\text{ef}+g_s^\text{in}}\Bigl]/(|\mathcal{S}|-k^*) -g_s^\text{in}\\
            & \le \frac{(g_s^\text{ef} + g_s^\text{in})}{|\mathcal{S}|-k^*} \cdot \Bigl[k^\text{ef} + k^\text{si}-  \sum_{s = k^\text{si} + k^\text{in}+1}^{k^*}\frac{\Gamma(k^*)+g_s^\text{in}}{g_s^\text{si}+g_s^\text{in}} -\sum_{s = 1}^{k^\text{si} + k^\text{in}}\frac{\Gamma(k^*)+g_s^\text{in}}{g_s^\text{ef}+g_s^\text{in}}\Bigl] -g_s^\text{in}\\
            & =   \frac{(g_s^\text{ef} + g_s^\text{in})}{|\mathcal{S}|-k^*} \cdot \Bigl[k^\text{ef} + k^\text{si}- \Gamma(k^*)\cdot H_{k^*} + \sum_{ s\le k^\text{si}+k^\text{in}}\frac{g_s^\text{si}}{g_s^\text{si} + g_s^\text{in}} -\frac{g_s^\text{si} + g_s^\text{in}}{g_s^\text{si} + g_s^\text{in}}  + \sum_{ k^\text{si}+k^\text{in}\le s\le k^* }\frac{g_s^\text{ef}}{g_s^\text{ef} + g_s^\text{in}} - \frac{g_s^\text{ef}+g_s^\text{in}}{g_s^\text{ef} + g_s^\text{in}}  \Bigl] -g_s^\text{in}\\
       &= \frac{(g_s^\text{ef} + g_s^\text{in})}{|\mathcal{S}|-k^*} \cdot [k^\text{ef} + k^\text{si}- (k^\text{ef}+k^\text{si}-|\mathcal{S}|+G_{k^*}) + G_{k^*} -k^*] -g_s^\text{in} = g_s^\text{ef} + g_s^\text{in} - g_s^\text{in} = g_s^\text{ef}
\end{align*}

According to the definition of $k^*$,  we have 
\begin{align*}
 &\frac{k^\text{ef}+k^\text{si}-|\mathcal{S}|+G_{k^*}}{H_{k^*}} = \Gamma(k^*) \ge \Gamma(k^*+1) \\
 &= \frac{k^\text{ef}+k^\text{si}-|\mathcal{S}|+G_{k^*}+g_{k^*+1}/(g_{k^*+1}+l_{k^*+1})}{H_{k^*}+1/(g_{k^*+1}+l_{k^*+1})}\\
 &\overset{(a)}{\ge}\frac{g_{k^*+1}/(g_{k^*+1}+l_{k^*+1})}{1/(g_{k^*+1}+l_{k^*+1})} = g_{k^*+1} \ge g_{s}  \text{ for } (s \ge k^* +1)
\end{align*}
where $(a)$ is established because if $a/b\ge(a+c)/(b+d)$ then $a/b\ge c/d$ for non-negative values.
Therefore, we have 
\begin{equation*}
    \Gamma(k^*)\ge g_s \ge t(c_{s}^\text{ef},c_{s}^\text{si},s) \text{ when } s\ge k^*.
\end{equation*}

To maximize $z(p, \hat{c})$, we apparently have $\sum_{s= 1}^{k^*} p_s= 1$ and $p_s=0$ when $s>k^*$. Therefore, we have 
\begin{equation*}
        \max_{p}\mathbb{E}_{s\sim p}[\hat{c}_{s}^\text{ef}\cdot g_s^\text{ef} + \hat{c}_{s}^\text{si}\cdot g_s^\text{si} - \hat{c}_{s}^\text{in}\cdot g_s^\text{in}] = \Gamma(k^*)
\end{equation*}
\end{proof}

\section{C. Detailed description of baselines}

We take the following state-of-the-art \textit{agnostic} KD methods as the baselines for comparison.
\begin{itemize}
    \item \textbf{DE} \cite{kang2020rrd} exploits “experts” with a selection strategy for distilling the knowledge of teacher’s representations to the student with limited capacity.
    \item \textbf{FTD} and \textbf{HTD} \cite{kang2021topology} are topology distillation approaches that transfer the Full/Hierarchical topological structure built upon the relations in the teacher space into the student.
    \item \textbf{unKD} \cite{chen2023unbiased} performs popularity-based debiasing KD during the training of the student model. 
    \item   \textbf{DSL} \cite{wang2024dynamic} trains a lightweight and sparse model, which periodically adjusts the significance of each weight and the model's sparsity distribution. 
    \item  \textbf{DLLM2Rec} \cite{cui2024distillation} is an LLM-based KD method based on importance-aware ranking distillation and collaborative embedding distillation. For a fair comparison in terms of efficiency, we extract instances with the same number as ours and remove confidence-aware weights that rely on generating for all items.
\end{itemize}

\begin{algorithm}
\caption{Algorithm of active learning for data selection. }
\label{alg:algorithm}
\textbf{Input}: Session set $\mathcal{S}$, item set $\mathcal{V}$, the well-trained conventional teacher $T_\text{llm}(\cdot)$, hyper-parameter $\mu$, and number of sampled instances $\tau$.
\begin{algorithmic}[1] 
\STATE \# \textbf{Difficulty and gains of each session measured by the conventional teacher $T_\text{llm}(\cdot)$}
\FOR{ $s \in \mathcal{S}$}
\STATE  $df_s = - \frac{1}{|s|}\sum\nolimits_{v\in s}\sigma(\text{Encode}(s)^T\cdot \boldsymbol{e}_{v})$ \quad\quad \# Difficulty.
\STATE $g_s^\text{ef} = 1/[rank(df_s)]^{\mu}$ \quad\quad\# Efficient gain ($g_s^\text{ef}$).
\STATE $g_s^\text{si}=g_s^\text{in}= g_s^\text{ef}/2$ \quad\quad\# Similarity gain ($g_s^\text{in}$) and ineffective gain ($-g_s^\text{in}$).
\ENDFOR
\STATE Ranking instances $\mathcal{S}$ based on their effective gains, i.e., $g_1^\text{ef}>g_2^\text{ef}>\cdots>g_{\vert\mathcal{S}\vert}^\text{ef}$.
\STATE \# \textbf{active instance selection policy $p^*$}
\STATE Calculate the $k^* = \text{argmax}_{s=1,\cdots,|\mathcal{S}|}\frac{g_s^\text{ef}+k^\text{ef}+k^\text{si}-|\mathcal{S}|}{H_s}$ in Theorem 2.
\STATE Calculate the active instance selection policy $p^* = (p_1^*,\cdots,p_{|\mathcal{S}|}^*)$, i.e.,
\begin{equation*}
p_s^*=\left\{
\begin{aligned}
&1/({H_{k^*}(g_s^\text{si}+g_s^\text{in})})  , &\text{if } 1 \le s\le k^\text{si} + k^\text{in}, \\
&1/({H_{k^*}(g_s^\text{ef}+g_s^\text{in})}) , &\text{if } k^\text{si} + k^\text{in}<s\le k^*,\\
&0  , &\text{if } s> k^*,\\
\end{aligned}
\right.
\end{equation*}
where 
\begin{equation*}
H_k=\left\{
\begin{aligned}
&\sum\nolimits_{s=1}^{k}\frac{1}{g_s^\text{si}+g_s^\text{in}}  , &\text{if } 1 \le k \le k^\text{si} + k^\text{in}, \\
&H_{k^\text{si} + k^\text{in}}+ \sum\nolimits_{s=1}^{k}\frac{1}{g_s^\text{ef}+g_s^\text{in}} , &\text{if } k^\text{si} + k^\text{in}< k \le k^*.\\
\end{aligned}
\right.
\end{equation*}
\STATE \# \textbf{Instance selection}
\STATE Initialize $\mathcal{B} = \emptyset$.
\WHILE{$|\mathcal{B}|<\tau$}
\STATE Sample $s\sim p^*$
\IF{$s \notin \mathcal{B}$}
\STATE $\mathcal{B} = \mathcal{B} + \{s\}$
\ENDIF
\ENDWHILE
\RETURN $\mathcal{B}$
\end{algorithmic}
\end{algorithm}

\end{document}